\documentclass[11pt]{article}

\usepackage{amssymb,amsmath,amsrefs,amsthm}
\usepackage[margin=1in]{geometry}
\usepackage{colonequals}
\usepackage{url}
\usepackage[bookmarksnumbered,colorlinks,linkcolor=blue,citecolor=blue,urlcolor=blue]{hyperref}
\usepackage{picture}
\usepackage{hypcap}

\renewcommand{\Im}{\operatorname{Im}}
\DeclareMathOperator{\interior}{int}
\DeclareMathOperator{\exterior}{ext}

\newtheorem{theorem}{Theorem}

\newtheorem{lemma}[theorem]{Lemma}

\newcommand{\eq}[1]{(\ref{eq:#1})}
\newcommand{\eqs}[2]{\eq{#1}--\eq{#2}}
\renewcommand{\sec}[1]{\hyperref[sec:#1]{Section~\ref*{sec:#1}}}
\newcommand{\thm}[1]{\hyperref[thm:#1]{Theorem~\ref*{thm:#1}}}
\newcommand{\lem}[1]{\hyperref[lem:#1]{Lemma~\ref*{lem:#1}}}
\newcommand{\prp}[1]{\hyperref[prop:#1]{Proposition~\ref*{prop:#1}}}
\newcommand{\fig}[1]{\hyperref[fig:#1]{Figure~\ref*{fig:#1}}}

\newcommand{\norm}[1]{\lVert{#1}\rVert}
\newcommand{\smfrac}[2]{{\textstyle\frac{#1}{#2}}}

\renewcommand{\d}{\mathrm{d}}
\newcommand{\ii}{\mathrm{i}}
\newcommand{\R}{\mathbb{R}}
\newcommand{\C}{\mathbb{C}}
\newcommand{\Z}{\mathbb{Z}}
\renewcommand{\S}{\mathcal{S}}

\newcommand{\defeq}{\colonequals}

\renewcommand{\>}{\rangle}
\newcommand{\<}{\langle}

\newcommand{\scattering}[1]{|\mathrm{sc}(#1)\>}

\newcommand{\bound}[2]{|\mathrm{bd}^{#1}(#2)\>}

\newcommand{\confined}{|\mathrm{cd}\>}

\begin{document}

\title{Levinson's theorem for graphs}
\author{
Andrew M.\ Childs\thanks{amchilds@uwaterloo.ca} \\[2pt]
Department of Combinatorics \& Optimization \\ 
and Institute for Quantum Computing \\
University of Waterloo
\and 
DJ Strouse\thanks{danieljstrouse@gmail.com} \\[2pt]
Department of Physics, University of Southern California \\
and Institute for Quantum Computing, University of Waterloo}
\date{}
\maketitle

\begin{abstract}
We prove an analog of Levinson's theorem for scattering on a weighted $(m+1)$-vertex graph with a semi-infinite path attached to one of its vertices.  In particular, we show that the number of bound states in such a scattering problem is equal to $m$ minus half the winding number of the phase of the reflection coefficient (where each so-called \emph{half-bound state} is counted as half a bound state).
\end{abstract}


\section{Introduction}

Levinson's theorem \cite{Lev49}, a classic result in quantum scattering theory, reveals a surprising connection between the scattering and bound states of a potential on a half line: the bound states of the potential are counted by the winding of the phase of the reflection coefficient.

Here we consider the setting of scattering on graphs.  In this model, scattering occurs on a discrete object, namely a (weighted) finite graph.  Semi-infinite paths are attached to some of the vertices, allowing for incoming and outgoing wave packets.  This setting was introduced (in the case of two semi-infinite paths) as a framework for developing algorithms for quantum computers \cite{FG98}.  Subsequently, scattering on graphs was used to discover a quantum algorithm for evaluating game trees \cite{FGG08}, resolving a longstanding open question in quantum query complexity.  The generalization to many semi-infinite paths was used to design a model of universal quantum computation \cite{Chi09}.

In fact, discrete analogs of quantum scattering theory were considered prior to the advent of quantum computation.  In the most natural discrete version of a potential on a half line, scattering occurs on a semi-infinite path of vertices with weighted self-loops representing the potential.  Case and Kac proved an analog of Levinson's theorem for this setting \cite{CK73}, and Hinton, Klaus, and Shaw generalized their result to apply to systems with bound states of a special kind, called \emph{half-bound states} \cite{HKS91}.

In this paper, we prove an analog of Levinson's theorem for scattering on an arbitrary weighted finite graph with a semi-infinite path attached to one of its vertices.  In particular, if the original finite graph contains $m$ vertices, we show that the number of bound states is equal to $m$ minus half the winding number of the phase of the reflection coefficient, where each half-bound state is counted as half a bound state.  The main idea of the proof is to consider the analytic continuation of the reflection coefficient from the unit circle to the rest of the complex plane.  Using the Cauchy argument principle, we relate the winding of the phase of the reflection coefficient around the unit circle to the zeros and poles of the analytically continued reflection coefficient inside the unit circle.  Finally, we relate these zeros and poles to bound states.

The remainder of the article is organized as follows.  In \sec{scatter} we introduce the model of scattering on graphs and describe scattering states and bound states.  In \sec{levinson} we prove the main result.  We conclude in \sec{discussion} with a discussion of possible directions for future work.


\section{Scattering on graphs}
\label{sec:scatter}

Consider an infinite path of vertices, each corresponding to a basis state $|x\>$ for $x \in \Z$, where vertex $x$ is connected to vertices $x \pm 1$.  The eigenstates of the adjacency matrix of this graph, parametrized by $k \in [-\pi,\pi)$, are the momentum states $|\tilde k\>$ with
\begin{align}
  \<x|\tilde k\> = e^{\ii k x}
\end{align}
for each $x \in \Z$.  These states are normalized so that $\<\tilde k|\tilde k'\> = 2\pi \, \delta(k-k')$.  Letting $H_{\mathrm{path}}$ denote the adjacency matrix of this infinite path, we have
\begin{align}
  \<x|H_{\mathrm{path}}|\tilde k\>
  &= \<x-1|\tilde k\> + \<x+1|\tilde k\> \\
  &= 2\cos k \, \<x|\tilde k\>, \label{eq:patheig}
\end{align}
which shows that $|\tilde k\>$ is an eigenstate of $H_{\mathrm{path}}$ with eigenvalue $2\cos k$.

Now let $G$ be an $(m+1)$-vertex graph, and create an infinite graph by attaching a semi-infinite path to the first vertex of $G$.  We refer to the remaining $m$ vertices of $G$ as \emph{internal vertices}.  Label the basis states for vertices on the semi-infinite path as $|x\>$, where $x=1$ for the vertex in the original graph and $x=2,3,\ldots$ moving out along the path.

In general, we may assign complex weights to the directed edges of $G$, subject to the Hermiticity constraint that edges in opposite directions have conjugate weights.  We denote the weighted adjacency matrix of $G$ by the block matrix
\begin{align}\label{eq:adjmat}
  \begin{pmatrix}
    a & b^\dag \\
    b & D
  \end{pmatrix}
\end{align}
where $a \in \R$, $b \in \C^m$ (a column vector), and $D = D^\dag \in \C^{m \times m}$.  A semi-infinite path is attached to the first vertex of $G$, with all edges on the semi-infinite path having weight $1$.  The adjacency matrix of the resulting infinite weighted graph is denoted $H$.  See \fig{graph} for a graphical representation of this Hamiltonian.  Note that the setting considered in Refs.~\cite{CK73,HKS91}, restricted to the case where the potential is only nonzero over a finite interval, is the special case where $b^\dag=(1,0,\ldots,0)$ and $D$ is a tridiagonal matrix with off-diagonal entries equal to $1$.

\begin{figure}
\capstart
\begin{center}
\setlength{\unitlength}{.1cm}
\begin{picture}(61,32)(-10,-6)
\thicklines
\multiput(20,10)(10,0){3}{\circle*{2}}
\put(20,10){\line(1,0){25}}
\multiput(47,10)(2,0){3}{\circle*{.5}}
\qbezier(20,10)(19,11)(19,13)
\qbezier(19,13)(19,15)(20,15)
\qbezier(21,13)(21,15)(20,15)
\qbezier(20,10)(21,11)(21,13)
\put(20,10){\line(-1,1){13}}
\put(20,10){\line(-3,2){13}}
\put(20,10){\line(-1,-1){13}}
\multiput(13,11)(0,-2){3}{\circle*{.5}}
\thinlines
\put(0,10){\oval(20,32)}
\put(19,16){\makebox(2,2){$a$}}
\put(13,18){\makebox(2,2){$b$}}
\put(-1,9){\makebox(2,2){$D$}}
\put(24,7){\makebox(2,2){$1$}}
\put(34,7){\makebox(2,2){$1$}}
\end{picture}
\end{center}
\caption{The scattering Hamiltonian.\label{fig:graph}}
\end{figure}
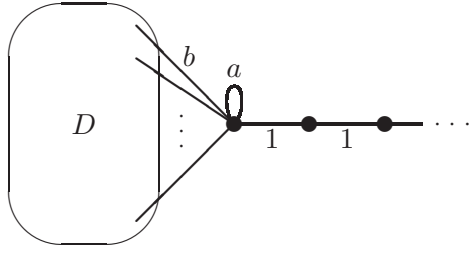

To analyze scattering on $G$, it is useful to determine the eigenstates of $H$.  Next we describe the two families of eigenstates, scattering states and bound states.

\subsection{Scattering states}
\label{sec:scattering}

For each $k \in \R \setminus \pi\Z$, there is a scattering state $\scattering{k}$ of momentum $k$.  This state has the form
\begin{align}\label{eq:scattering}
  \<x \scattering{k} &= e^{-\ii k x} + R(e^{\ii k}) \, e^{\ii k x}
\end{align}
on the semi-infinite path, where $R(e^{\ii k})$ is called the reflection coefficient.  We can view $R$ as a function of $e^{\ii k}$ since it is $2\pi$-periodic in $k$.  By a similar calculation as in \eq{patheig}, such states satisfy the eigenvalue condition
\begin{align}\label{eq:eigcondsc}
  H \scattering{k} = 2 \cos k \, \scattering{k}  
\end{align}
at vertices of the semi-infinite path.  Requiring \eq{eigcondsc} to hold at the $m+1$ vertices of $G$ gives a system of $m+1$ linear equations that determines the reflection coefficient $R(e^{\ii k})$ and the amplitudes of $\scattering{k}$ at the $m$ internal vertices of $G$.

Given the Hamiltonian $H$ defined in terms of the weighted adjacency matrix \eq{adjmat}, it is straightforward to solve for the reflection coefficient.  The eigenvalue conditions for the scattering states $\scattering{k}$ at the vertices of $G$ can be collected into the vector equation
\begin{align}\label{eq:eigeqn}
  \begin{pmatrix} a & b^\dag \\ b & D \end{pmatrix}
  \begin{pmatrix}e^{-\ii k}+R(e^{\ii k}) \, e^{\ii k} \\ \psi(e^{\ii k})\end{pmatrix}
  +\begin{pmatrix}e^{-2\ii k}+R(e^{\ii k}) \, e^{2\ii k} \\ 0\end{pmatrix}
  &=
  2\cos k \begin{pmatrix}e^{-\ii k}+R(e^{\ii k}) \, e^{\ii k} \\ \psi(e^{\ii k})\end{pmatrix},
\end{align}
where $\psi(e^{\ii k}) \in \C^m$ is the vector of amplitudes at the $m$ internal vertices of $G$ in the state $\scattering{k}$.  The lower block gives
\begin{align}
  \psi(e^{\ii k}) = (e^{-\ii k}+R(e^{\ii k}) \, e^{\ii k}) (2\cos k - D)^{-1} b.
\end{align}
Applying this identity to the upper block of \eq{eigeqn} and solving for $R(e^{\ii k})$, we find\footnote{Note that the expression \eq{rcoeff} depends on the choice \eq{scattering} for the form of the scattering states.  In particular, it assumes a particular convention for the overall phases of these states.  While the stated form will turn out to be convenient for our purposes, other natural choices can be obtained by taking states of the same form, but numbering vertices on the semi-infinite path starting from an integer other than $1$.  Such a choice modifies the reflection coefficient by factors of $e^{\ii k}$, which ultimately modifies the statement of Levinson's theorem.}
\begin{align}\label{eq:rcoeff}
  R(e^{\ii k}) = -\frac{Q(e^{-\ii k})}{Q(e^{\ii k})}
\end{align}
where
\begin{align}
  Q(e^{\ii k}) &\defeq 1 - e^{\ii k}(a+C(e^{\ii k})) \label{eq:qmatrix} \\
  C(e^{\ii k}) &\defeq b^\dag (2\cos k - D)^{-1} b. \label{eq:c}
\end{align}
It is easy to see that $|R(e^{\ii k})|=1$ since $Q(e^{\ii k})^* = Q(e^{-\ii k})$.

The definition of $R(e^{\ii k})$ as a parameter describing scattering states only makes sense for $k \in \R$.  Nevertheless, it is convenient to extend the definition of $R$ to the whole complex plane.  For all $z \in \C$ we let
\begin{align}\label{eq:rcomplex}
  R(z) \defeq - \frac{Q(z^{-1})}{Q(z)}
\end{align}
where
\begin{align}
  Q(z) &\defeq 1 - z(a + C(z)) \label{eq:qcomplex} \\
  C(z) &\defeq b^\dag(z+z^{-1} - D)^{-1}b. \label{eq:ccomplex}
\end{align}
With $z = e^{\ii k}$, these definitions agree with \eqs{rcoeff}{c}.  In fact, $R(z)$ is the analytic continuation of $R(e^{\ii k})$ from the unit circle $\Gamma \defeq \{z \in \C \colon |z|=1\}$ to $\C$.  Note that for $z \notin \Gamma$, we may have $|R(z)| \ne 1$.

\subsection{Bound states}
\label{sec:bound}

Scattering states alone may not span the entire state space.  In general, to obtain a complete basis, we must also include bound states.

For $\kappa \ge 0$, a state $\bound{\pm}{\kappa}$ of the form
\begin{align}
  \<x \bound{\pm}{\kappa} = (\pm e^{-\kappa})^x
\end{align}
on the semi-infinite path satisfies
\begin{align}\label{eq:eigcondbd}
  H \bound{\pm}{\kappa} = \pm 2 \cosh\kappa \, \bound{\pm}{\kappa}
\end{align}
at vertices of the semi-infinite path.  If some choice of the amplitudes at the $m$ internal vertices of the graph also satisfies the condition \eq{eigcondbd} on the vertices of $G$, then $\bound{\pm}{\kappa}$ is an eigenstate of $H$ with eigenvalue $\pm2\cosh\kappa$.  The $m$ eigenvalue conditions at the internal vertices of $G$ give linear equations for the $m$ amplitudes at internal vertices in terms of $\kappa$.  The remaining eigenvalue condition at the first vertex of $G$ gives a transcendental equation in $\kappa$ with a discrete set of solutions, each corresponding to a bound state.

The states $\bound{\pm}{\kappa}$ can be characterized as follows.  Such a state satisfies the vector equation
\begin{align}\label{eq:eigeqnbd}
  \begin{pmatrix} a & b^\dag \\ b & D \end{pmatrix}
  \begin{pmatrix}\pm e^{-\kappa} \\ \psi^\pm(\kappa) \end{pmatrix}
  +\begin{pmatrix} e^{-2\kappa} \\ 0\end{pmatrix}
  &=
  \pm 2\cosh\kappa \begin{pmatrix}\pm e^{-\kappa} \\
  \psi^\pm(\kappa) \end{pmatrix}
\end{align}
where $\psi^\pm(\kappa) \in \C^m$ is a vector of amplitudes at the $m$ internal vertices of $G$.  The lower block gives
\begin{align}
  \psi^\pm(\kappa) = e^{-\kappa} (2\cosh\kappa \mp D)^{-1} b .
\end{align}
Using this in the upper block of \eq{eigeqnbd} gives $Q(\pm e^{-\kappa}) = 0$, where $Q(z)$ is given in \eq{qcomplex}.  In other words, bound states correspond to values of $\kappa$ for which $\pm e^{-\kappa}$ is a root of $Q(z)$.

The states $\bound{\pm}{\kappa}$ with $\kappa>0$ have amplitudes that decay exponentially along the semi-infinite path, so we refer to them collectively as \emph{evanescent bound states}.  However, in the preceding discussion, we have allowed the possibility that $\kappa=0$.  States of the form $\bound{\pm}{0}$ may or may not exist depending on whether $Q(\pm 1) = 1 \mp (a + C(\pm 1)) = 0$.  Using the terminology of Ref.~\cite{HKS91}, we refer to such states as \emph{half-bound states}.  This reflects that they are, in a sense, halfway between scattering and bound states: when they exist, we have $\bound{+}{0} = \scattering{0}$ and $\bound{-}{0} = \scattering{\pi}$.
It also reflects the role they play in Levinson's theorem.

Finally, there may be bound states with zero amplitude on the semi-infinite path.  We refer to such states as \emph{confined bound states} (they have also been called \emph{bound states of the second kind} in the terminology of Ref.~\cite{VB09}).  A state $\confined$ with zero amplitude on the semi-infinite path satisfies the eigenvalue equation $H\confined = \lambda\confined$ on the vertices of the semi-infinite path for any value of $\lambda$.  To satisfy the eigenvalue equation at the internal vertices of $G$, a state $\confined$ must satisfy the vector equation
\begin{align}
  \begin{pmatrix} a & b^\dag \\ b & D \end{pmatrix}
  \begin{pmatrix}0 \\ \psi\end{pmatrix}
  &=
  \lambda \begin{pmatrix}0 \\ \psi\end{pmatrix}.
\end{align}
The lower block says that $\psi$ is an eigenvector of $D$ with eigenvalue $\lambda$, and the upper block says that $b^\dag \psi=0$.  Thus, confined bound states correspond to eigenvectors $\psi$ of $D$ that also satisfy the further condition of orthogonality to $b$.  Note that if $\lambda$ is a degenerate eigenvalue of $D$, then there may be multiple confined bound states corresponding to $\lambda$.  In general, the number of confined bound states corresponding to $\lambda$ is the dimension of the subspace $\{\psi\colon D\psi = \lambda\psi,\, b^\dag \psi = 0\}$.  If $b$ is orthogonal to the $\lambda$-eigenspace of $D$, this dimension is simply the multiplicity of $\lambda$; otherwise it is the multiplicity of $\lambda$ minus $1$.

It can be shown that the scattering states $\{\scattering{k}\colon k \in (-\pi,0)\}$, together with the bound states $\{\bound{\pm}{\kappa}\colon \kappa \ge 0,\; Q(\pm e^{-\kappa}) = 0\}$
and confined bound states
$\{\confined\}$, 
form a complete basis for the state space \cite{Gol08}.


\section{Levinson's theorem}
\label{sec:levinson}

In order to state Levinson's theorem, we first introduce the concept of the winding number of a complex-valued function around a curve in the complex plane.  Let $\gamma$ be a simple, closed, positively-oriented (i.e., counterclockwise) curve in $\C$, and let $f\colon \C \to \C$ be a meromorphic function with no zeros or poles on $\gamma$.
Then we can write $f(z) = r(z) \, e^{\ii \theta(z)}$ for some functions $r\colon \C \to \R$ and $\theta\colon \C \to \R$, where $0 < r(z) < \infty$ for all $z \in \gamma$.  For each $z \in \C$, $\theta(z)$ is only defined up to an integer multiple of $2\pi$.  However, we can require $\theta(z)$ to be continuous as we vary $z \in \gamma$, in which case $\theta(z)$ is determined up to an overall offset by an integer multiple of $2\pi$.  Then for $f(z)$ to be single-valued, the function $\theta(z)$ must change by some integer multiple of $2\pi$ as $z$ winds around $\gamma$.  This integer is the \emph{winding number of $f$ around $\gamma$},
\begin{align}
  w_\gamma(f)
  &\defeq \frac{1}{2\pi\ii} \oint_\gamma \frac{f'(z)}{f(z)} \d{z} \\
  &=      \frac{1}{2\pi} \oint_\gamma \d{\theta(z)}.
\end{align}

We can compute $w_\gamma(f)$ by appealing to the \emph{Cauchy argument principle}, which says that the winding number of $f$ around $\gamma$ is determined by the zeros and poles of $f$ inside $\gamma$.  In particular, we have the following:

\begin{theorem}[Cauchy argument principle]\label{thm:argprinciple}
Let $\gamma$ be a simple, closed, positively-oriented curve in $\C$, and let $f\colon \C \to \C$ be a meromorphic function with no zeros or poles on $\gamma$.  Then
\begin{align}
  w_\gamma(f) = Z_{\interior(\gamma)}(f)-P_{\interior(\gamma)}(f)
\end{align}
where $Z_{\interior(\gamma)}(f)$ is the number of zeros of $f$ inside $\gamma$, counted with their multiplicity, and $P_{\interior(\gamma)}(f)$ is the number of poles of $f$ inside $\gamma$, counted with their order.
\end{theorem}

Levinson's theorem relates the winding of the reflection coefficient around the unit circle $\Gamma$ to the number of bound states.  In particular, we have the following:

\begin{theorem}[Levinson's theorem for graphs]\label{thm:levinson}
Let $R(z)$ be the reflection coefficient of a graph scattering problem with $m$ internal vertices in which there are $n_c$ confined bound states, $n_h$ half-bound states, and $n_b$ evanescent bound states.  Furthermore, suppose that either $a \ne 0$ or $\norm{b} \ne 1$ (or both).  Then
\begin{align}
  w_\Gamma(R) = 2 (m - n_b - n_c - \smfrac{1}{2}n_h).
\end{align}
\end{theorem}

The technical condition on $a$ and $b$ poses no essential difficulty.  If $a=0$ and $\norm{b}=1$, then by a suitable change of basis, the problem is equivalent to one in which $a=0$ and $b^\dag=(1,0,\ldots,0)$.  In this case, we can effectively view the first vertex of the finite graph as part of the semi-infinite path, and we can replace the $(m+1) \times (m+1)$ matrix \eq{adjmat} by the $m \times m$ matrix $D$.  Repeating this process, we eventually arrive at a graph with $a \ne 0$ or $\norm{b} \ne 1$ (or else we arrive at a graph with a single vertex, in which case the bound states are easily characterized).

The proof of \thm{levinson} uses the argument principle.  Note that since $|R(z)|=1$ for $|z|=1$, $R(z)$ has no zeros or poles on the unit circle.  Our goal is to count the zeros and poles of $R(z)$ inside the unit circle.

It will be helpful to view $R(z)$ as a quotient as in \eq{rcomplex}
and consider the numerator and denominator of this expression separately.  Indeed, except at the point $z=0$, every zero of $R(z)$ inside $\Gamma$ corresponds to a zero of $Q(z^{-1})$, and every pole of $R(z)$ inside $\Gamma$ corresponds to a zero of $Q(z)$.

\begin{lemma}\label{lem:numdenom}
Let $0 < |z| < 1$.  Then $R(z)=0$ if and only if $Q(z^{-1})=0$, and $R(z)^{-1} = 0$ if and only if $Q(z)=0$.
\end{lemma}

\begin{proof}
Suppose that $R(z) = 0$.  We claim that $Q(z^{-1})=0$; assume for a contradiction that this is not the case.  To satisfy $R(z)=0$, $Q(z)$ must be unbounded.  Since $a$ is a constant and $|z| < 1$, this can only happen if $C(z)$ is unbounded.  However, if $C(z)$ is unbounded then $R(z) = -z^{-2} \ne 0$.  This is a contradiction, so we must have $Q(z^{-1}) = 0$.

Similarly, if $R(z)^{-1} = 0$ and $Q(z) \ne 0$, then $Q(z^{-1})$ must be unbounded, and since $a$ and $z^{-1}$ are finite, $C(z^{-1}) = C(z)$ must again be unbounded, showing that $R(z)^{-1} = -z^2 \ne 0$.  This is a contradiction, showing that $Q(z) = 0$.

Clearly, if $Q(z^{-1})=0$ and $Q(z) \ne 0$ then $R(z)=0$, and if $Q(z)=0$ and $Q(z^{-1}) \ne 0$ then $R(z)^{-1}=0$.  It remains to consider the possibility that $Q(z)=Q(z^{-1})=0$.  In this case, $z(a+C(z)) = z^{-1}(a+C(z^{-1})) = z^{-1}(a+C(z))$.  Since $a+C(z) \ne 0$ (as $a+C(z)=0$ implies $Q(z)=Q(z^{-1})=1 \ne 0$), this shows that $z^2=1$, which violates the condition $|z|<1$.
\end{proof}

Thus it suffices for us to count the zeros of $Q(z^{-1})$ and of $Q(z)$ inside $\Gamma$.  First we count the zeros of $Q(z)$ inside $\Gamma$.  We saw in \sec{bound} that every bound state corresponds to a root of $Q(z)$.  In particular, for any $z = \pm e^{-\kappa} \in {[-1,0)} \cup {(0,1]}$ with $Q(z) = 0$, we can solve \eq{eigeqnbd} to find a bound state.  Next we show that zeros of $Q(z)$ on and inside the unit circle can only occur on the real axis, which demonstrates that every zero of $Q(z)$ on and inside the unit circle (except at $z=0$) corresponds to a bound state.

\begin{lemma}\label{lem:boundstatesreal}
  Let $|z| \le 1$ and $Q(z)=0$.  Then $\Im(z)=0$.
\end{lemma}

\begin{proof}
  The claim is trivial for $z=0$, so we can assume $|z|>0$.  Then
\begin{align}
  0 
  &= \Im(z^{-1} Q(z)) \\
  &= \Im(z^{-1}) - \Im(C(z)).
\end{align}
With $z = re^{i\theta}$, we have $\Im(z^{-1}) = -r^{-1} \sin\theta$.  Let $D = \sum_\lambda \lambda P_\lambda$ be the spectral decomposition of $D$, where $P_\lambda$ projects onto the eigenspace of $D$ with eigenvalue $\lambda$.  Then we have
\begin{align}\label{eq:cdecomp}
  C(z)
  &= \sum_\lambda \frac{b^\dag P_\lambda b}{z+z^{-1}-\lambda}.
\end{align}
Therefore,
\begin{align}
  \Im(C(z))
  &= - \sum_\lambda \frac{\norm{P_\lambda b}^2 (r-r^{-1})\sin\theta}{[(r+r^{-1})\cos\theta - \lambda]^2 + [(r-r^{-1})\sin\theta]^2}.
\end{align}
Now assume for a contradiction that $\sin\theta \ne 0$.  Then we have
\begin{align}
  r^{-1} = (r-r^{-1}) \sum_\lambda \frac{\norm{P_\lambda b}^2}{[(r+r^{-1})\cos\theta - \lambda]^2 + [(r-r^{-1})\sin\theta]^2}.
\end{align}
For $0 < r \le 1$, the left hand side is positive but the right hand side is non-positive.  This is a contradiction, so we must have $\sin\theta=0$, i.e., $\Im(z)=0$.
\end{proof}

Furthermore, each bound state contributes a zero of multiplicity one.

\begin{lemma}\label{lem:norepeatedroots}
Suppose that $Q(x)=0$ where $0<|x|\le 1$.  Then $\frac{\d}{\d{x}} Q(x) \ne 0$.
\end{lemma}

\begin{proof}
By \lem{boundstatesreal}, $x \in \R$.  Since $Q(x)=0$, $a+C(x) = x^{-1}$.  Then we have
\begin{align}
  \frac{\d}{\d{x}} Q(x)
  &= -(a+C(x)) - x \frac{\d}{\d{x}}C(x) \\
  &= -x^{-1} - x\frac{\d}{\d{x}} \sum_\lambda \frac{\norm{P_\lambda b}^2}{x+x^{-1}-\lambda} \\
  &= -x^{-1} - \sum_\lambda \frac{\norm{P_\lambda b}^2 (x^{-1}-x)}{(x+x^{-1}-\lambda)^2},
\end{align}
which is strictly negative for $0<x\le 1$ and strictly positive for $-1\le x<0$.
\end{proof}

So far, we have ignored the point $z=0$.  Although there are not bound states at  $z=0$, $R(z)$ can nevertheless have a pole at $z=0$.

\begin{lemma}\label{lem:zero}
If $a \ne 0$, then $R(z)$ has a simple pole at $z=0$.  If $a=0$ and $\norm{b} \ne 1$, then $R(z)$ has no zeros or poles at $z=0$.
\end{lemma}

\begin{proof}
We have $R(z) = -Q(z^{-1}) / Q(z)$.
Note that $C(0)=0$.
Since $Q(0)=1$, $Q(z)$ contributes no zeros or poles at $z=0$.  Now $Q(z^{-1}) = z^{-1} (z-a-C(z))$.  If $a \ne 0$, then $\lim_{z \to 0} z-a-C(z) = -a \ne 0$, so $R(z)$ has a simple pole at $z=0$ as claimed.  On the other hand, if $a=0$, then $\lim_{z \to 0} Q(z^{-1}) = \lim_{z \to 0} (1-z^{-1} C(z)) = 1-\norm{b}^2$.  Thus if $a=0$ and $\norm{b} \ne 1$, $R(z)$ has no zeros or poles at $z=0$ as claimed.
\end{proof}

We must also count the zeros of $Q(z^{-1})$ inside $\Gamma$.  We do this by relating the zeros of $Q(z^{-1})$ inside $\Gamma$ to the zeros of $Q(z)$ outside $\Gamma$.  Observe that if $Q(z_{\mathrm{in}}^{-1})=0$ where $|z_{\mathrm{in}}|<1$, then $Q(z_{\mathrm{out}})=0$ for some $z_{\mathrm{out}}$ with $|z_{\mathrm{out}}|>1$, namely $z_{\mathrm{out}} = z_{\mathrm{in}}^{-1}$.  Thus, it will be sufficient to determine the total number of zeros of $Q(z)$.

To count the zeros of $Q(z)$, we use the fact that it is a rational function of $z$ and determine the degree of its numerator.  First we consider the case where $a \ne 0$.

\begin{lemma}\label{lem:degree}
If $a \ne 0$, then $Q(z)$ is a rational function of $z$ whose numerator has degree $2m - 2n_c + 1$.
If $a = 0$ and $\norm{b} \ne 1$, then $Q(z)$ is a rational function of $z$ whose numerator has degree $2m - 2n_c$.
\end{lemma}

\begin{proof}
Using \eq{cdecomp}, we have
\begin{align}
  Q(z) 
  &= 1 - az - z^2 \sum_\lambda \frac{\norm{P_\lambda b}^2}{z^2-\lambda z+1} \\
  &= \frac{(1-az)\prod_\lambda (z^2 - \lambda z + 1) - z^2 \sum_\lambda \norm{P_\lambda b}^2 \prod_{\lambda' \ne \lambda} (z^2 - \lambda' z + 1)}{\prod_\lambda (z^2 - \lambda z + 1)}.
  \label{eq:commondenom}
\end{align}

If $\lambda = \pm 2$, then $z^2-\lambda z+1 = (z \mp 1)^2$ has two identical roots ($z = \pm 1$).  For $\lambda \ne \pm 2$, we claim that the two roots of $z^2-\lambda z+1$ are distinct. Furthermore, for $\lambda \ne \lambda'$, $z^2-\lambda z+1$ and $z^2-\lambda' z+1$ have no common roots.
To see this, write $z^2-\lambda z+1 = (z-\mu_+)(z-\mu_-)$ where
\begin{align}
  \mu_\pm &\defeq \smfrac{\lambda}{2} \pm \sqrt{\big(\smfrac{\lambda}{2}\big)^2-1}.
\end{align}
Observe that if $|\lambda| < 2$ then $\mu_+$ is on the half of the unit circle in the complex plane with positive imaginary part and $\mu_-$ is on the half of the unit circle with negative imaginary part.  If $\lambda > 2$ then $\mu_+ \in (1,\infty)$ and $\mu_- \in (0,1)$, and if $\lambda < -2$ then $\mu_+ \in (-1,0)$ and $\mu_- \in (-\infty,-1)$.  Thus, clearly $\mu_+ \ne \mu'_-$ regardless of whether $\lambda = \lambda'$.
We also claim that $\mu_+ \ne \mu'_+$ and $\mu_- \ne \mu'_-$ when $\lambda \ne \lambda'$.
For $|\lambda|<2$ this is obvious, simply by comparing real parts.
For $|\lambda|>2$, observe that $\mu_+$ is a strictly increasing function of $\lambda$ for $\lambda>2$ and a strictly decreasing function for $\lambda<-2$, and $\mu_-$ is a strictly decreasing function of $\lambda$ for $\lambda>2$ and a strictly increasing function for $\lambda<-2$; therefore, for $\lambda \ne \lambda'$, $\mu_+ \ne \mu'_+$ and $\mu_- \ne \mu'_-$.

Now let $J(z)$ denote the numerator of \eq{commondenom}, so that
\begin{align}
  Q(z) &= \frac{J(z)}{\prod_\lambda (z^2-\lambda z + 1)}.
\end{align}
If $a \ne 0$, then $J(z)$ is a polynomial with leading term $-az^{2\bar m+1}$, where $\bar m$ is the number of distinct eigenvalues of $D$.  On the other hand, if $a=0$ and $\norm{b} \ne 1$, then the leading term of $J(z)$ is $(1 - \sum_\lambda \norm{P_\lambda b}^2) z^{2 \bar m} = (1-\norm{b}^2) z^{2 \bar m}$.

It remains to count common factors between the numerator and denominator of \eq{commondenom}.  For an eigenvalue $\lambda \ne \pm 2$ of $D$, let $\mu \ne \pm 1$ be either of the two roots of $z^2-\lambda z + 1$.  Then $\mu$ is a simple root of $\prod_{\lambda'} (z^2 - \lambda' z + 1)$, and
\begin{align}
  \label{eq:Jmu}
  J(\mu) = - \mu^2 \norm{P_\lambda b}^2 \prod_{\lambda' \ne \lambda} (\mu^2-\lambda' \mu + 1),
\end{align}
so $\mu$ is a root of $J(z)$ if and only if $\norm{P_\lambda b}=0$.
If $\lambda = \pm 2$ is an eigenvalue of $D$, then $\mu = \pm 1$ is a double root of $\prod_{\lambda'} (z^2 - \lambda' z + 1)$.  If $\norm{P_\lambda b} \ne 0$ then \eq{Jmu} shows that $\mu$ is not a root of $J(z)$, and if $\norm{P_\lambda b} = 0$ then it is easy to see that $\mu$ is a root of $J(z)$ of multiplicity at least $2$.  Thus we see that the number of common factors between the numerator and denominator of \eq{commondenom} is twice the number of values $\lambda$ such that $\norm{P_\lambda b} = 0$.  Since the number of confined bound states corresponding to $\lambda$ is either the multiplicity of $\lambda$ (if $\norm{P_\lambda b} = 0$) or the multiplicity of $\lambda$ minus $1$ (if $\norm{P_\lambda b} \ne 0$), the total number of common factors is $2(n_c + \bar m - m)$.  Then the degree of the numerator of $Q(z)$ is $2m - 2n_c + 1$ if $a \ne 0$ and $2m - 2n_c$ if $a = 0$ and $\norm{b} \ne 1$, as claimed.
\end{proof}

Finally, we are ready to prove the main result.

\begin{proof}[Proof of \thm{levinson}]
For any $\S \subseteq \C$ and any meromorphic function $f$, let $Z_\S(f)$ denote the number of zeros of $f$ that belong to $\mathcal{S}$ (counted with their multiplicity) and let $P_\S(f)$ denote the number of poles of $f$ that belong to $\S$ (counted with their order).  Let $\interior(\Gamma)$ denote the interior of $\Gamma$, and let $\exterior(\Gamma)$ denote the exterior of $\Gamma$ (both open sets, not including points on $\Gamma$).  Since $R(z)$ is a rational function of $z$, in particular it is a meromorphic function of $z$; since $|R(z)|=1$ for $|z|=1$, it has no zeros or poles on $\Gamma$.  Thus, by the argument principle (\thm{argprinciple}),
\begin{align}
  w_\Gamma(R)=Z_{\interior(\Gamma)}(R)-P_{\interior(\Gamma)}(R).
\end{align}
By \lem{zero}, $R(z)$ has a simple pole at $z=0$ if $a \ne 0$ and no zeros or poles at $z=0$ otherwise, so
\begin{align}
w_\Gamma(R)=Z_{\interior(\Gamma)\setminus\{0\}}(R)-P_{\interior(\Gamma)\setminus\{0\}}(R)-\delta[a \ne 0]
\end{align}
where $\delta[a \ne 0]=1$ if $a \ne 0$ and $\delta[a \ne 0]=0$ if $a = 0$.

By \lem{numdenom},
\begin{align}
  Z_{\interior(\Gamma)\setminus\{0\}}(R)
  &= Z_{\interior(\Gamma)\setminus\{0\}}(Q(z^{-1})) \\
  \text{and} \quad
  P_{\interior(\Gamma)\setminus\{0\}}(R) &= Z_{\interior(\Gamma)\setminus\{0\}}(Q(z)).
\end{align}
In particular,
\begin{align}
  w_\Gamma(R)=Z_{\interior(\Gamma)\setminus\{0\}}(Q(z^{-1})) - Z_{\interior(\Gamma)\setminus\{0\}}(Q(z)) - \delta[a \ne 0].
\end{align}

By \lem{boundstatesreal} and \lem{norepeatedroots}, $Z_{\interior(\Gamma) \setminus \{0\}}(Q(z)) = Z_{\interior(\Gamma)}(Q(z)) = n_b$.
Also, as discussed above, we can relate zeros of $Q(z^{-1})$ inside $\Gamma$ to zeros of $Q(z)$ outside $\Gamma$.  In particular,
\begin{align}
  Z_{\interior(\Gamma)\setminus\{0\}}(Q(z^{-1}))
  &= Z_{\exterior(\Gamma)}(Q(z)) \\
  &= Z_{\C}(Q(z)) - Z_{\Gamma}(Q(z)) - Z_{\interior(\Gamma)}(Q(z)).
\end{align}

To compute $Z_{\Gamma}(Q(z))$, recall that zeros of $Q(z)$ at $z=\pm1$ correspond to half-bound states.  By \lem{boundstatesreal}, these are the only zeros with $|z|=1$, and by \lem{norepeatedroots}, they occur with multiplicity one.  Thus we have $Z_{\Gamma}(Q(z)) = n_h$, the total number of half-bound states.

Finally, by \lem{degree}, $Z_{\C}(Q(z)) = 2m - 2n_c + \delta[a \ne 0]$.
Combining these expressions, we have
\begin{align}
  w_\Gamma(R)=2m-2n_b-2n_c-n_h
\end{align}
as claimed.
\end{proof}


\section{Discussion}
\label{sec:discussion}

We conclude by discussing several possible directions for future work.

First, it might be interesting to investigate possible applications of bound states to quantum computation.  Existing quantum algorithms based on graph scattering \cite{FGG08,Chi09} rely solely on the properties of scattering states, but one can imagine that bound states could also encode information about the solution of a computational problem.  While it might be possible to extract information about bound states directly using phase estimation, it could also be useful to infer such information indirectly via Levinson's theorem.

One could also consider generalizing \thm{levinson} to the case where $n > 1$ semi-infinite paths are attached to $n$ vertices of an $(n+m)$-vertex graph.  In general, the scattering is described by an $n \times n$ matrix, the $S$-matrix, instead of a single reflection coefficient (see Ref.~\cite{Chi09} for details).  Similar calculations to those in \sec{scattering} can be used to determine the $S$-matrix in this more general case: for a weighted adjacency matrix of the form $\left(\begin{smallmatrix}A & B^\dag \\ B & D\end{smallmatrix}\right)$, where $A \in \C^{n \times n}$, $B \in \C^{m \times n}$, and $D \in \C^{m \times m}$, the $S$-matrix at momentum $k$ is $S(e^{\ii k})$, where $S(z) = -Q(z^{-1})/Q(z)$ with $Q(z) = I - z(A + B^\dag[(z + z^{-1}) I - D]^{-1}B)$.  Furthermore, similar conditions to those described in \sec{bound} can be used to characterize the bound states.  We expect the number of bound states in such a scattering problem to be related to the winding number of $\det S(z)$, but we leave the details as a topic for future work.

Finally, the original theorem of Levinson for a potential on a half-line not only describes a relationship between the number of bound states and the phase of the reflection coefficient, but also sheds light on the extent to which the potential can be inferred from the behavior of the reflection coefficient.  It might be fruitful to consider such inverse problems in the setting of scattering on graphs.  The weighted graph corresponding to a given form of the reflection coefficient is certainly far from unique, since one may perform arbitrary similarity transformations on the $m$ internal vertices without affecting the scattering behavior, but one might attempt to characterize the graph corresponding to a given reflection coefficient up to such changes of basis.


\section*{Acknowledgments}
We thank Jeffery Goldstone for sharing his proof of completeness of scattering and bound states \cite{Gol08} and Gorjan Alagic, Aaron Denney, and Cris Moore for discussions of methods for computing $S$-matrices of graphs.
We also thank an anonymous referee for pointing out several corrections to an earlier version of this paper.
This work was supported in part by MITACS, NSERC, QuantumWorks, and the US ARO/DTO.


\end{document}